\documentclass[aps,prl,twocolumn]{revtex4}

\usepackage{subfiles}
\usepackage{amssymb,amsmath}
\usepackage{bm, graphicx, amsmath}
\usepackage{bbm}
\usepackage[section]{placeins}
\usepackage{tikz}
\usepackage{multirow}
\usepackage{ytableau}
\ytableausetup{boxsize=0.8em}
\usepackage[mathscr]{eucal}
\usepackage{graphicx}
\usepackage{color}
\usepackage{enumitem}
\usepackage{float}
\usepackage{hyperref}
\usepackage{xr-hyper}

\newcommand{\tr}{\operatorname{tr}}
\newcommand{\Span}{\operatorname{span}}

\newcommand{\two}[2]{\begin{array}{c}\\[-1.5em]\scriptstyle #1\\[-.3em] \scriptstyle #2\end{array}}

\newcommand{\Hc}{\mathcal{H}}
\newcommand{\Hs}{\mathscr{H}}

\newcommand{\Gc}{\mathcal{G}}
\newcommand{\Xc}{X}

\newcommand{\Ss}{\mathscr{S}}

\newcommand{\stab}{\operatorname{stab}}
\newcommand{\cen}{\operatorname{cen}}

\newcommand{\co}[1]{#1\kern-.05em\Hc}
\newcommand{\xbf}{\boldsymbol{x}}

\renewcommand{\r}{r}

\newtheorem{theo}{Theorem}
\newtheorem{lem}{Lemma}
\newtheorem{prop}{Proposition}
\newtheorem{cor}{Corollary}

\begin{document}


\parskip=.7em

\title{Quantum-Enhanced Zero-Error Communication and Storage under Positional Uncertainty}
\author{A. Diebra$^{1}$, D. Gonz\'alez-Lociga$^{1}$, M. Hillery$^{2,3}$, J. Calsamiglia$^{1}$, and  E. Bagan$^{1}$}

\affiliation{$^{1}$F\'{i}sica Te\`{o}rica: Informaci\'{o} i Fen\`{o}mens Qu\`antics, Universitat Aut\`{o}noma de Barcelona, 08193 Bellaterra (Barcelona), Spain
$^{2}$Department of Physics and Astronomy, Hunter College of the City University of New York, 695 Park Avenue, New York, NY 10065, USA\\
$^{3}$Physics Program, Graduate Center of the City University of New York, 365 Fifth Avenue, New York, New York 10016, USA
}

\begin{abstract}
Permutation channels model communication and storage scenarios in which the positional identity of the physical carriers is partially or completely lost, so that the transmitted information is only accessible up to an unknown reordering. Here we show that quantum mechanics can dramatically enhance zero-error communication through such channels. For cyclic reorderings of $n$ $d$-level systems, and in the absence of positional metadata, the number of classical zero-error messages scales asymptotically as $d^n/n$, whereas quantum protocols can fully recover the identity-channel value~$d^n$. Ancilla-assisted protocols further increase this number to $d^{2n}/n$, enabling dense coding under positional uncertainty. We also analyze dihedral permutation channels and derive general Pólya-like formulas for the number of distinguishable messages in a broad class of permutation groups. Finally, for the symmetric group $S_n$, corresponding to complete scrambling of the information carriers, the number of distinguishable messages scales as $n^{d-1}$ classically, compared with $n^{d(d+1)/2-1}$ for quantum protocols and $n^{d^2-1}$ in the ancilla-assisted setting. Our results establish a fundamental quantum advantage for communication and storage under positional uncertainty.
\end{abstract}

\maketitle

In many communication and storage systems, information is not only vulnerable to noise and losses, but also to a distinct form of uncertainty: the loss of positional identity. Messages may be fragmented, shuffled, displaced, or retrieved in an order unrelated to the one in which they were originally encoded. These situations can be naturally modeled by {\it permutation channels}~\cite{walsh2007coding,langberg2015coding,kovavcevic2014zero}, where the transmitted data is only accessible up to an unknown reordering of its constituent parts. As modern technologies increasingly rely on large collections of mobile, distributed, or indistinguishable carriers, understanding how much information can still be reliably transmitted or stored under such uncertainty becomes both a practical and theoretical challenge.

Permutation channels arise naturally whenever positional information is degraded or lost. A particularly important example is DNA-based archival storage~\cite{heckel2017fundamental,sabary2024survey}, where data is encoded into a large collection of short DNA strands that are later retrieved in an essentially unordered fashion through sequencing, requiring substantial resources to recover their original arrangement. Similar issues also appear in packet-based networks~\cite{feng2009packet,lin2023packet}, asynchronous transmission protocols~\cite{tchamkerten2012asynchronous}, and molecular communication~\cite{Nakano_Eckford_Haraguchi_2013,gohari2016information}. Such scenarios naturally raise questions about the ultimate limits of communication under positional uncertainty, particularly in the form of permutation uncertainty.

In practice, positional uncertainty is typically mitigated by incorporating additional metadata, such as sequence numbers, timestamps, synchronization markers, or explicit addressing information~\cite{schwartz2012accurate}. These mechanisms effectively restore ordering information, but inevitably consume communication or storage resources and thereby reduce the effective transmission rate. In DNA-based storage, for instance, a non-negligible fraction of each strand must be devoted to indexing information in order to reconstruct the original ordering after sequencing~\cite{heckel2017fundamental}. More generally, the need to recover positional identity introduces an intrinsic overhead in communication protocols operating under permutation uncertainty.

These limitations naturally motivate the exploration of genuinely quantum communication strategies for permutation channels. Throughout this work, the carriers are not supplemented with positional metadata, since this would enlarge the local dimension and partially break the symmetry underlying the permutation channels themselves. Under the same assumptions, quantum mechanics nonetheless allows such reference information to be encoded into coherent superpositions and relative phases, opening the possibility of qualitatively different communication protocols. It is then natural to ask whether these quantum effects can provide an advantage for zero-error communication and storage in such constrained scenarios.

For permutation uncertainty, the possible reorderings naturally define a permutation group acting on the set of length-$n$ strings over an alphabet of size $d$ used to encode messages. Classical messages are then identified up to the corresponding group orbits. Consequently, the number of perfectly distinguishable classical messages grows much more slowly than the total number $d^n$ of possible strings, often leading to vanishing asymptotic zero-error rates~\cite{sarkar2025identification}.

Quantum mechanics does not, in general, remove this obstruction or restore a finite asymptotic rate. But coherent superpositions can substantially enlarge the number of messages that can be retrieved with certainty. For example, under a cyclic permutation channel, the classical number scales asymptotically as $d^n/n$, whereas quantum encodings can recover the full identity-channel value $d^n$. Ancilla-assisted strategies increase this further to $d^{2n}/n$. Under the full symmetric group $S_n$, the corresponding classical number grows only polynomially, as~$n^{d-1}$, while the numbers achievable with quantum and ancilla-assisted encodings scale respectively as~$n^{d(d+1)/2-1}$ and $n^{d^2-1}$.

Accordingly, our focus is not on asymptotic transmission rates, but on the single-shot zero-error number of messages that can be transmitted or stored with certainty, and on the symmetry-based mechanisms that make this possible.  The loss of positional identity is not merely ordinary noise: it endows the set of possible messages with an underlying symmetry structure. For certain permutation groups, these mechanisms can be realized through quantum Fourier-type encodings over the corresponding orbits, making the effect both conceptually transparent and potentially accessible to current experimental platforms~\cite{bluvstein2022quantum}.

Related questions concerning communication under incomplete or missing reference information have previously appeared, for example, in the context of quantum communication without a shared reference frame~\mbox{\cite{bartlett2003classical,bartlett2007reference,bartlett2009quantum}}. Permutations have also been considered as controlled encoding resources in entanglement-assisted communication protocols~\cite{wang2021permutation}. Here, however, we focus on the fundamentally different setting in which the ordering of the physical carriers is itself inaccessible.

This Letter is organized as follows. We first introduce a general framework for communication and storage under permutation uncertainty. We then specialize to cyclic positional uncertainty, illustrated by a ring of $n$ cold atoms with $d$ internal states used as a memory device, where the atoms may collectively rotate around the ring. In this setting, the classical problem reduces to necklace counting~\cite{vanLintWilson2001,harary1973graphical}. We next allow the ring to flip over, leading to the dihedral group $D_n$, and derive the corresponding classical and quantum asymptotic behaviors. The analysis is then extended to other physically motivated permutation groups. Finally, we consider full permutation uncertainty, corresponding to the symmetric group $S_n$.

\noindent{\em General framework.} As mentioned in the introduction, messages are encoded into strings of symbols, which may be pictured as balls with $d$ possible colors. In the quantum setting, the symbols are replaced by quantum systems with $d$ orthonormal internal states. A permutation group $\Gc\le S_n$ (i.e., a subgroup of the symmetric group~$S_n$) acts on the classical strings \mbox{$\xbf=x_0x_1\dots x_{n-1}$} by permuting their positions according to \mbox{$\sigma\xbf:=x_{\sigma^{-1}(0)}x_{\sigma^{-1}(1)}\cdots x_{\sigma^{-1}(n-1)}$, for $\sigma\in\Gc$}. Likewise, in the quantum setting, the action of $\sigma$ defines a unitary representation through \mbox{$U(\sigma)|\xbf\rangle:=|\sigma\xbf\rangle$}, where $|\xbf\rangle$ is shorthand for the computational basis state \mbox{$|x_0\rangle\otimes|x_1\rangle\otimes\cdots\otimes|x_{n-1}\rangle$}. We denote by $X$ the set of all such strings or, by a slight abuse of notation, the corresponding computational-basis vectors.

Recall that the $\Gc$-orbit, or simply orbit, of a string \mbox{$\xbf\in X$} is the set \mbox{$O_{\xbf}=\{\sigma\xbf:\sigma\in\Gc\}$}, namely the strings obtained from $\xbf$ through the action of the group. In the quantum setting, we use the same notation for the set \mbox{$\{U(\sigma)|\xbf\rangle:\sigma\in\Gc\}$}. Strings belonging to the same orbit therefore cannot encode distinct classical messages~through the $\Gc$-channel. The number of classical messages that can be transmitted with certainty is thus simply the number $N_{\rm c}$ of orbits in $X$, namely $N_{\rm c}=|X/\Gc|$ in standard mathematical notation. Since orbits are equivalence classes, they partition $X$ into disjoint subsets.

Burnside's lemma~\cite{vanLintWilson2001,supp_mat} states that
\begin{equation}
N_{\rm c}={1\over|\Gc|}\sum_{\sigma\in\Gc}| X_\sigma|
={1\over |\Gc|}\sum_{\sigma\in\Gc} d^{c(\sigma)},
\label{burnside}
\end{equation}
where $ X_\sigma$ denotes the fixed-point set of $\sigma$, i.e., the set of strings in $ X$ left invariant by $\sigma$. The second equality is the simplest instance of Pólya's enumeration theorem~\cite{harary1973graphical}, where $c(\sigma)$ denotes the number of cycles in the disjoint cycle decomposition of $\sigma$. It follows immediately from the identity $| X_\sigma|=d^{c(\sigma)}$, since a string is fixed by $\sigma$ if and only if all positions belonging to the same cycle of $\sigma$ carry the same symbol. One may therefore choose the symbol associated with each cycle independently, yielding $d^{c(\sigma)}$ invariant strings.

In the quantum setting, the Hilbert space of the information carriers decomposes into orthogonal subspaces associated with the different orbits,
\begin{equation}
(\mathbb{C}^d)^{\otimes n}=\bigoplus_j \Span(O_{\xbf_j}),
\end{equation}
where $\{\xbf_j\}_{j=0}^{N_{\rm c}-1}$ is a set of representatives, one from each orbit. 

Every orbit subspace further decomposes into irreducible representations
(irreps) of $\Gc$~\cite{Serre1977},
\begin{equation}
\Span(O_{\xbf_j})
=
\bigoplus_{\mu}\Hs_\mu\otimes\mathbb{C}^{m_{j,\mu}},
\end{equation}
where $\Hs_\mu$ denotes the irrep labeled by~$\mu$, appearing with
multiplicity~$m_{j,\mu}$. The spaces $\mathbb{C}^{m_{j,\mu}}$ are the
multiplicity spaces, on which $\Gc$ acts trivially. By definition of
irreducibility, no nontrivial proper subspace of $\Hs_\mu$ is invariant under
the action of $\Gc$. Thus, within each of the $m_{j,\mu}$ copies of
$\Hs_\mu$, one cannot select two or more linearly independent states that
remain perfectly distinguishable under the unknown group action. Such a copy
therefore encodes exactly one perfectly distinguishable message, which can be
decoded by a projective measurement onto~the corresponding copy of~$\Hs_\mu$.
We denote by $\{|u^\mu_\alpha\rangle\}_{\alpha=0}^{m_\mu-1}$ the orthogonal
encoding states gathered from the different orbit subspaces and call the
subspace they generate the ``message space''. Hence, the number of perfectly
distinguishable messages that can be encoded without ancillary systems is
\begin{equation}
N_{\rm q}
=
\sum_\mu\Bigg(\sum_j m_{j,\mu}\Bigg)
=
\sum_\mu m_\mu,
\label{M_q}
\end{equation}
i.e., the sum of all multiplicities.

If an ancillary system of dimension $m_\mu$ is available and is not affected by the permutation channel, this message space can be used for dense coding~\cite{Bennett1992,Werner2001}. Let $\{|u^\mu_\alpha\rangle\}_{\alpha=0}^{m_\mu-1}$ be 
the orthonormal set of encoding vectors, and let $\{|\varphi^\mu_\alpha\rangle\}_{\alpha=0}^{m_\mu-1}$ be an orthonormal basis of the ancilla. The sender and receiver may share the maximally entangled state
\begin{equation}
|\Phi_\mu\rangle
=
{1\over\sqrt{m_\mu}}
\sum_{\alpha=0}^{m_\mu-1}
|u^\mu_\alpha\rangle\otimes|\varphi^\mu_\alpha\rangle .
\end{equation}
By applying a complete set of $m_\mu^2$ orthogonal unitary transformations on the message space, the sender can generate $m_\mu^2$ mutually orthogonal maximally entangled states. Since the ancillary system is unaffected by the permutation channel, these states remain perfectly distinguishable at the receiver. Thus, the sector associated with $\mu$ contributes $m_\mu^2$ distinguishable messages, and the ancilla-assisted number of messages is
\begin{equation}
N_{\rm a}=\sum_\mu m_\mu^2
=
{1\over |\Gc|} \sum_{\sigma\in\Gc} \left|\chi(\sigma)\right|^2,
\end{equation}
where $\chi(\sigma)$ is the character of the representation $U(\sigma)$ on $(\mathbb{C}^d)^{\otimes n}$, and we have used the standard representation-theoretic identity relating the squared multiplicities to the absolute value of the character. Moreover, \mbox{$\chi(\sigma)=\sum_{\xbf\in X}\langle\xbf|U(\sigma)|\xbf\rangle=|X_\sigma|=d^{c(\sigma)}$}, and we therefore obtain the compact expression
\begin{equation}
N_{\rm a}= {1\over |\Gc|} \sum_{\sigma\in\Gc} d^{2c(\sigma)}.
\label{N_a polya}
\end{equation}
Hence, dense coding over a permutation channel allows the transmission of $N_{\rm a}$ mutually distinguishable messages. This is the ultimate limit permitted by quantum mechanics under permutation uncertainty.

\noindent{\em Cyclic group.} To illustrate the above formalism, we consider the simplest nontrivial setting: a cyclic permutation group. As a concrete example, we revisit the ring of cold atoms introduced above as a memory device. In this scenario, the carriers of information are arranged on a circle, and their relative ordering is well defined, but there is no distinguished reference position. In other words, the stored message is defined only up to an unknown cyclic shift of the atomic positions.
 
Since the symmetry group $\Gc$ in this example is cyclic, and therefore abelian, all its irreps are one-dimensional~\cite{Serre1977}, i.e., \mbox{$\dim(\Hs_\mu)=1$} for all~$\mu$. Using \mbox{$\sum_\mu m_\mu \dim(\Hs_\mu)=\dim[(\mathbb{C}^d)^{\otimes n}]=d^n$}, Eq.~(\ref{M_q}) immediately gives
\begin{equation}
N_{\rm q}=d^n.
\label{cyclic_general}
\end{equation}
%


The corresponding classical and ancilla-assisted expressions follow from the
cycle structure of the powers of the permutation
\mbox{$r=(0,1,\dots,n-1)$}, representing a one-step shift and written in
standard cycle notation, which generates
\mbox{$\Gc=\langle r\rangle\simeq\mathbb Z_n$}. The powers of $r$ take the
form $r^k=(0,k,2k,\dots)$, with indices understood modulo~$n$. Thus, if $p$
is the smallest positive integer such that $pk=0\pmod n$, then $r^k$
decomposes into $n/p$ cycles of length~$p$. Elementary modular arithmetic
gives $p=n/\gcd(k,n)$, where $\gcd$ denotes the greatest common divisor.
Hence, $c(r^k)=\gcd(k,n)$ (see~\cite{supp_mat} for a worked-out example).
%
%
%
%
%
Equations~(\ref{burnside}) and~(\ref{N_a polya}) then yield
\begin{equation}
N_{\rm c}
={1\over n}\sum_{k=0}^{n-1}d^{\gcd(k,n)};
\qquad
N_{\rm a}
={1\over n}\sum_{k=0}^{n-1}d^{2\gcd(k,n)}.
\label{Nc&Na}
\end{equation}
The first quantity coincides with the number of distinct necklaces that can be formed with $n$ beads of $d$ colors, a well-known problem in classical combinatorics.
In the asymptotic regime $n\to\infty$, the dominant contribution comes from the term $k=0$, for which $\gcd(0,n)=n$, leading to
\begin{equation}
N_{\rm c}\sim {d^n\over n};
\qquad
N_{\rm a}\sim {d^{2n}\over n}.
\end{equation}
In other words, without ancillary systems, quantum encoding allows a factor-$n$ increase in the number of perfectly distinguishable classical messages compared with the corresponding classical protocol under the same symmetry constraints, fully recovering the identity-channel value. Ancillary systems further give an additional factor-$d^n$ enhancement.

The orthonormal basis $\{|u^\mu_\alpha\rangle\}_{\alpha=0}^{m_\mu-1}$ of the message space, consisting of vectors that transform under $\mathbb{Z}_n$ according to its characters [i.e., phases $(\omega_n)^\mu$, with~\mbox{$\omega_n={\rm e}^{2\pi i/n}$}], can be constructed by applying the quantum Fourier transform (QFT)~\cite{shor1999polynomial,nielsen2010quantum} along the orbits. To this end, we first determine a set of orbit representatives~$\xbf_j$, for instance using the Fredricksen--Kessler--Maiorana (FKM) algorithm~\cite{FredricksenMaiorana1978}, which generates them efficiently in lexicographical order. Denoting by $n_j=|O_{\xbf_j}|$ the size of the orbit of $\xbf_j$, we define the Fourier basis within each orbit as
\begin{equation}
|\tilde u^k_j\rangle=\frac{1}{\sqrt{n_j}}\sum_{l=0}^{n_j-1} (\omega_{n_j})^{-kl} U(\r^l)|\xbf_j\rangle,
\label{ebc26.03.26-1}
\end{equation}
%
%
%
where $k\in\{0,1,\dots,n_j-1\}$. These states diagonalize the action of~$\r$,
$
U(\r)|\tilde u^k_j\rangle
=
(\omega_{n_j})^k |\tilde u^k_j\rangle.
$
By the orbit--stabilizer theorem~\cite{harary1973graphical,supp_mat}, $n_j$ divides~$n$, so that $n=n_jp_j$ for some integer~$p_j$. Then,
$
\omega_{n_j}
=
(\omega_n)^{p_j},
$
and we obtain
$
U(\r)|\tilde u^k_j\rangle
=
(\omega_n)^{p_jk} |\tilde u^k_j\rangle.
$
Each $|\tilde u^k_j\rangle$ therefore transforms according to a one-dimensional irreducible representation (character) of~$\mathbb{Z}_n$, namely the one labeled by~$\mu=p_jk$. Collecting, from all $N_{\rm c}$ orbits, the states with the same value of~$\mu$ yields the desired basis~$\{|u^\mu_\alpha\rangle\}_{\alpha=0}^{m_\mu-1}$.


\noindent{\em Dihedral group $D_n$.} Assume now that the ring of~$n$ cold atoms not only drifts rigidly but can also be flipped. This enlarges the symmetry group, which is now generated by two elements: $r$, with $r^n=e$, and $s$, with $s^2=e$ and $srs=r^{-1}$. The resulting group, \mbox{$D_n=\{e,r,r^2,\dots ,r^{n-1},s,sr,sr^2,\dots ,sr^{n-1}\}$}, of order $2n$, is the dihedral group: the group of symmetries of a regular polygon with $n$ sides. The elements $sr^k$ correspond to reflections. For even $n$, these reflections occur along axes passing through opposite vertices or opposite sides, whereas for odd $n$ each reflection axis passes through one vertex and the midpoint of the opposite side.

The number of distinguishable classical messages (the bracelet-counting problem~\mbox{\cite{vanLintWilson2001,harary1973graphical}}) and its ancilla-assisted counterpart can be computed from Eqs.~(\ref{burnside}) and~(\ref{N_a polya}):
\begin{equation}
N_{\rm c}
=
{1\over 2n}
\left[\,
\sum_{k=0}^{n-1} d^{\gcd(k,n)}
+
f(n)\,
\right],
\end{equation}
with
\begin{equation}
f(n)=
\begin{cases}
{n\over2}d^{{n\over2}}(d+1), & n~\mbox{even},\\[.4em]
n d^{{n+1\over2}}, & n~\mbox{odd}.
\end{cases}
\end{equation}
The contribution $f(n)$ arises from the reflections, while the sum accounts for the cyclic subgroup $\langle r\rangle\trianglelefteq D_n$, as in Eq.~(\ref{Nc&Na}). The corresponding ancilla-assisted result, $N_{\rm a}$, is obtained by replacing $d$ with $d^2$. For asymptotically large $n$, the contribution of $f(n)$ is subdominant, and
\begin{equation}
N_{\rm c}\sim {d^n\over 2n};
\qquad
N_{\rm a}\sim {d^{2n}\over 2n}.
\label{Dn asy}
\end{equation}

Before deriving the corresponding ancilla-free expression, we reformulate Eq.~(\ref{M_q}) for an important family of groups that includes $D_n$ and the symmetric group $S_n$.

\noindent{\em Message counting within ancilla-free quantum protocols for totally orthogonal groups.} Except for special cases such as cyclic groups, the computation of $N_{\rm q}$ becomes much more involved. However, if {\em all\,} irreps of $\Gc$ are realizable over the real numbers, one obtains Pólya-like formulas analogous to Eqs.~(\ref{burnside}) and~(\ref{N_a polya}):
\begin{equation}
N_{\rm q}={1\over|\Gc|}\sum_{\sigma\in\Gc}d^{c(\sigma^2)}.
\label{Nq polya}
\end{equation}
Such groups, known as totally orthogonal, include many physically relevant examples, notably the dihedral group $D_n$ and the symmetric group $S_n$.

The proof of Eq.~(\ref{Nq polya}) can be found in~\cite{supp_mat}.~It~exploits the fact that, for every irrep $\mu$, the Frobenius--Schur indicator of totally orthogonal groups is equal to $+1$~\cite{james2001representations}:
\mbox{$\nu_\mu:=|\Gc|^{-1}\sum_{\sigma\in\Gc} \chi_\mu(\sigma^2)=1$},
where $\chi_\mu(\sigma)$ is the character of the irrep $\mu$ of $\Gc$.

For the dihedral group, Eq.~(\ref{Nq polya}) gives
\begin{equation}
N_{\rm q}={d^n\over2}+{1\over2n}\sum_{k=0}^{n-1}d^{\gcd(b_n k,n)},
\end{equation}
where $b_n=1$ for odd $n$ and $b_n=2$ for even $n$. The first term arises from reflections, since $(sr^k)^2=e$, whose cycle decomposition consists of $n$ $1$-cycles. The second term comes from rotations, for which $\sigma^2=r^{2k}$. Asymptotically, the reflection contribution dominates, and therefore
\begin{equation}
N_{\rm q}\sim {d^n\over 2}.
\end{equation}
Comparing with Eq.~(\ref{Dn asy}), we recover the expected hierarchy \mbox{$N_{\rm c}<N_{\rm q}<N_{\rm a}$} for $n,d\ge 2$, providing a consistency check for the expressions above.

As in the cyclic case, the quantum encoding for the dihedral group can be implemented using the QFT~\cite{prep}, making this protocol experimentally accessible with current quantum technology~\cite{bluvstein2022quantum}.

\noindent{\em The symmetric group $S_n$.} As a concrete example, let us assume that messages are transmitted asynchronously, for instance due to diffusion, as in molecular communication, or because of routing processes in communication networks. We model this situation with a full permutation channel, i.e., a channel whose symmetry is described by the symmetric group $S_n$.

The number $N_{\rm c}$ of messages that can be transmitted in the classical setting can be computed from Eq.~(\ref{burnside}). However, a more direct derivation uses the fact that under arbitrary permutations only the occupation numbers of the different symbols are preserved. A straightforward stars-and-bars combinatorial argument yields
\begin{equation}
N_{\rm c}={n+d-1\choose n},
\qquad
N_{\rm a}={n+d^2-1\choose n},
\end{equation}
where the ancilla-assisted expression is obtained from the classical one by the replacement $d\mapsto d^2$. Asymptotically,
\begin{equation}
N_{\rm c}\sim {n^{d-1}\over (d-1)!},
\qquad
N_{\rm a}\sim {n^{d^2-1}\over (d^2-1)!}.
\label{Nc Na Sn}
\end{equation}
We note that the exponential scaling found for cyclic and dihedral symmetries is lost under such a demanding channel, becoming merely polynomial. Even ancilla-assisted encoding cannot restore the exponential behavior, although it still provides a significant improvement over the classical setting.

Without ancillas, the number of perfectly distinguishable messages is given by Eq.~(\ref{Nq polya}). The value of $c(\sigma^2)$ depends on the parity of the cycle lengths: a $k$-cycle of $\sigma$ remains a $k$-cycle under squaring when $k$ is odd, whereas for even $k$ it splits into two cycles of length $k/2$. Therefore,
\begin{equation}
N_{\rm q}={1\over n!}\sum_{\sigma\in S_n}\;
\prod_{k~{\rm odd}} d^{c_k(\sigma)}
\prod_{k~{\rm even}} d^{2c_k(\sigma)},
\end{equation}
where $c_k(\sigma)$ is the number of $k$-cycles in the decomposition of $\sigma$. This expression can be written equivalently as
\mbox{$
N_{\rm q}=Z_{S_n}(d,d^2,d,d^2,\dots),
$}
where
\begin{equation}
Z_{S_n}(a_1,a_2,\dots,a_n):=
{1\over n!}
\sum_{\sigma\in S_n}
\prod_{k=1}^n a_k^{c_k(\sigma)}
\end{equation}
is the cycle index function of $S_n$~\cite{polya2012combinatorial} in the independent variables $a_1,a_2,\dots,a_n$. Its generating function, commonly known as the cycle index series, has the compact form
\begin{equation}
\sum_{n=0}^\infty x^n Z_{S_n}(a_1,a_2,\dots,a_n)
=
\exp\!\left(\,\sum_{k=1}^\infty {a_k x^k\over k}\right).
\end{equation}
For the case of interest here, $a_{2j+1}=d$ and $a_{2j}=d^2$, the series in the exponent can be evaluated in closed form, yielding
\begin{equation}
N_{\rm q}
=
[x^n]\!\left[
(1-x)^{-{d(d+1)\over2}}
(1+x)^{-{d(d-1)\over2}}
\right],
\end{equation}
where $[x^n]f(x)$ denotes the coefficient of $x^n$ in the Taylor expansion of $f(x)$. The asymptotic behavior is then obtained from the dominant singularity at $x=1$:
\begin{equation}
N_{\rm q}\sim
{n^{{d(d+1)\over2}-1}
\over
2^{d(d-1)\over2}
\left[{d(d+1)\over2}-1\right]!}.
\end{equation}
Together with Eq.~(\ref{Nc Na Sn}), this reveals a clear quantum advantage over the classical encoding for the most general scrambling of information carriers.

The explicit construction of the encoding states and additional details can be found in~\cite{prep}, which also addresses the transmission and storage of quantum information.

\noindent{\em Summary and outlook.} We have shown that quantum mechanics can substantially enhance zero-error communication and storage when the positional identity of the information carriers is lost. Without additional positional information, classical protocols suffer a severe reduction in the number of distinguishable messages, whereas quantum coherence and entanglement allow information to be encoded into collective states of the carriers. For cyclic rearrangements, this restores the communication and storage capabilities of an ideal ordered array. In addition, we demonstrated a clear quantum advantage for both dihedral and fully random scrambling. In the cyclic and dihedral cases, the corresponding protocols can be implemented through QFT, making them particularly appealing for near-term experimental realizations. Dense-coding-like enhancements remain possible even under positional uncertainty.

Future work should address the robustness of these protocols under realistic noise and the development of concrete proof-of-principle experimental implementations. In particular, recent advances in coherent transport of neutral-atom arrays while preserving internal quantum states suggest that these protocols may already be realizable in current programmable quantum platforms~\cite{bluvstein2022quantum}. More broadly, our results suggest a new paradigm for communication and storage in architectures where the ordering of the carriers is intrinsically uncertain or dynamically lost, including mobile quantum platforms and molecular or DNA-based information systems.

\noindent\textit{Acknowledgments}. We thank A. Ac\'{\i}n for valuable discussions and encouragement. 
This work was supported by MCIN with funding from the European Union NextGenerationEU \mbox{(PRTR-C17.I1)}, the Generalitat de Catalunya, and the Spanish MINECO-TD through the QUANTUM ENIA project ``Quantum Spain'', funded by the European Union Recovery, Transformation and Resilience Plan – NextGenerationEU within the framework of the ``Digital Spain 2026 Agenda''. Additional support was provided by grant PID2022-141283NB-I00 (MICIU/AEI/10.13039/501100011033). J.C. acknowledges support from ICREA Academia, and A.D. from MICIN through grant FPU23/02763. M.H. is supported by the National Science Foundation under the grant Collaborative Research: NeTS: Medium 2504622.



\bibliography{bibliography} 

\cleardoublepage 

\setcounter{section}{0} 
\setcounter{equation}{0} 
\setcounter{figure}{0} 
\setcounter{table}{0} 
\setcounter{page}{1}

\parskip=.7em

\title{SUPPLEMENTAL MATERIAL FOR\\
Quantum-Enhanced Zero-Error Communication and Storage under Positional Uncertainty}

\author{A. Diebra$^{1}$, D. Gonz\'alez-Lociga$^{1}$, 
M. Hillery$^{2,3}$, J. Calsamiglia$^{1}$, and E. Bagan$^{1}$}

\affiliation{
$^{1}$F\'{i}sica Te\`{o}rica: Informaci\'{o} i Fen\`{o}mens Qu\`antics, 
Universitat Aut\`{o}noma de Barcelona, 08193 Bellaterra (Barcelona), Spain
$^{2}$Department of Physics and Astronomy, Hunter College of the City University 
of New York, 695 Park Avenue, New York, NY 10065, USA\\
$^{3}$Physics Program, Graduate Center of the City University of New York, 
365 Fifth Avenue, New York, New York 10016, USA
}
\begin{abstract}

\end{abstract}

\pacs{03.67.-a, 03.65.Ta,42.50.-p }
\maketitle 
\onecolumngrid

All equations in this supplementary note are numbered with the prefix `S'.
Equations referenced without this prefix correspond to those in the main
text. The sections are organized according to the order in which they are
used in the main text, rather than thematically.

\section{Group actions, the orbit--stabilizer theorem, and Burnside's lemma}

The material in this section is standard and can be found in textbooks on
combinatorics~\cite{vanLintWilson2001,polya2012combinatorial} and group
theory~\cite{Serre1977,james2001representations}. For the reader's convenience, we summarize the main
tools used to derive the number of messages that can be transmitted using
only classical resources. Proofs of selected statements, such as the
orbit-stabilizer theorem and Burnside's lemma, are also provided.

Before introducing the main group-theoretical notions used throughout these
notes, we briefly recall some elementary definitions from set theory, which
also serve to establish notation and conventions.

\subsubsection{Equivalence classes and quotient sets}

Let $\sim$ be an equivalence relation on a set $\Xc$. The {\em equivalence
class} of $x\in\Xc$ is the subset
$$
[x]=\{y\in\Xc\;|\; y\sim x\}.
$$
The element $x$ is called a {\em representative} of the class $[x]$. The
choice of representative is not unique: every element $y\in[x]$ can be chosen
as a representative of $[x]$.

The set of equivalence classes of $\sim$ is called the {\em quotient set}
and is denoted by $\Xc/\sim$.
Any equivalence relation partitions the set $\Xc$ into disjoint classes.
Therefore, if $\Xc$ is finite,
$$
|\Xc|=\sum_{j=1}^N |[x_j]|,
$$
where $|Y|$ denotes the number of elements (cardinality) of a set $Y$,
$N=|\Xc/\sim|$, and $x_j$ is a representative of the $j$th equivalence class.

\subsubsection{Left and right cosets. Normal subgroups}

Let $\Hc$ be a subgroup of a group $\Gc$, which we denote by $\Hc\le \Gc$.
We define a relation on $\Gc$: $g \sim g'$ if $g' = g h$ for some
$h \in \Hc$. This is an equivalence relation.

All elements equivalent to $g$ form the {\em left coset} of $g$:
$\co g=\{g'\in\Gc : g'\sim g\}$. Hence, left cosets are a special case of
equivalence classes in which the underlying set, $\Gc$, carries a group
structure. As before, we call $g$ a representative~of the coset, although
any element of the coset can equally well serve as a representative.
The notation $\co g$ reflects the fact that its elements are obtained by
multiplying $g$ on the left by elements of $\Hc$. In particular, if $\Gc$
is finite, every left coset contains the same number of elements as $\Hc$.
As for any equivalence relation, the one we have defined partitions the
group $\Gc$ into disjoint left cosets. Therefore, if $\Gc$ is finite,
$|\Hc|$ divides $|\Gc|$, i.e.,
$$
|\Gc| / |\Hc|\in \mathbb{N}.
$$
One can similarly define right cosets of the form $\Hc g$, with analogous
properties. In general, $\Hc g \neq \co g$ unless $\Hc$ is a normal subgroup.

By definition, a subgroup $\Hc\le\Gc$ is a {\em normal subgroup}, denoted
$\Hc\trianglelefteq\Gc$, if it is invariant under conjugation, meaning that
for every $h\in\Hc$ and every $g\in\Gc$, one has $ghg^{-1}\in\Hc$.
Equivalently, for every $g\in\Gc$ and $h\in\Hc$, there exists $h'\in\Hc$
such that $gh=h'g$. In particular, for normal subgroups one has
$\Hc g=\co g$.

In general, the quotient set, denoted $\Gc/\Hc$ (instead of $\Gc/\sim$), is
not a group unless $\Hc$ is normal. In that case, one can define the product
of cosets by
$$
(\co g)(\co{g'}) = \co{(gg')}.
$$
This operation is well defined precisely because $\Hc$ is normal.

\subsubsection{Action of a group on a set}\label{action_sss}

\noindent{\em An action of a group $\Gc$ on a set $\Xc$} is a mapping
$$
\varphi:\Gc\times\Xc\rightarrow\Xc,
$$
satisfying
\begin{enumerate}[label={(\alph*)}]
\item $\varphi(e,x)=x$ for all $x\in\Xc$,
\item $\varphi(g',\varphi(g,x))=\varphi(g'g,x)$ for all
$g,g'\in\Gc$ and $x\in\Xc$.
\end{enumerate}
It is convenient to introduce the shorthand notation
$gx:=\varphi(g,x)$. The above conditions then become
\begin{enumerate}[label={(\alph*)}]
\item $ex=x$ for all $x\in\Xc$,
\item $(g'g)x=g'(gx)$ for all $g,g'\in\Gc$ and $x\in\Xc$.
\end{enumerate}
The concept of a group action is ubiquitous in Physics. Space-time
transformations, permutations of particles, and time evolution are all
examples of group actions. In these examples, as well as in many other
applications, the set $\Xc$ carries additional mathematical structure beyond
that of a simple set, such as that of a vector space, group, or algebra.

\noindent{\em The stabilizer of $x\in\Xc$} is defined as
$$
\stab(x)=\left\{g\in\Gc\;|\; gx=x\right\}.
$$
In words, it is the set of group elements that leave $x$ invariant.
It is convenient to introduce the shorthand notation
$\Gc_x:=\stab(x)$. We will use both notations interchangeably throughout
these notes. One can easily check that the stabilizer is a subgroup of
$\Gc$.

\noindent{\em The fixed-point set of $g\in\Gc$} is the subset
$\Xc_g\subset\Xc$ of elements fixed by $g$, namely
$$
\Xc_g=\{x\in\Xc\;|\; gx=x\}.
$$

\noindent{\em The $\Gc$-orbit of $x\in\Xc$} (or simply orbit of $x$ if no ambiguity arises) is defined as the set
$$
O_x=\left\{x'\in\Xc\;|\; x'=gx,\; g\in\Gc\right\},
$$
or simply
$$
O_x=\left\{gx\;|\; g\in\Gc\right\},
$$
of all elements of $\Xc$ that can be reached from $x$ by the action of
elements of the group.
The orbits $O_x$ are precisely the equivalence classes of the relation
$$
x_1\sim x_2 \;\Leftrightarrow\; x_2=gx_1
\ \text{for some } g\in\Gc.
$$
Hence,
$$
[x_1]
=
\left\{x_2\in\Xc\;|\; x_2=gx_1
\ \text{for some } g\in\Gc\right\}
=
O_{x_1},
$$
and the set of all these equivalence classes, or orbits, is denoted by
$\Xc/\Gc$. This implies that
\begin{enumerate}[label={(\alph*)}]
\item Any two orbits are either identical or disjoint.
\item For a finite set $\Xc$,
$$
|\Xc|
=
\sum_{j=1}^N |O_{x_j}|,
$$
where $O_{x_j}$, $j=1,2,\dots,N$, are the
$N=|\Xc/\Gc|$ distinct orbits of $\Xc$. The elements $x_j$ are
representatives of the corresponding orbits.
\end{enumerate}

\subsubsection{Orbit-stabilizer theorem}

\begin{theo}\label{ebc21.03.26-t1}
Let the group $\Gc$ act on the set $\Xc$. Then the mapping
$
\phi:\Gc/\Gc_x \rightarrow O_x,
$
defined by
$$
\phi(g\Gc_x)=gx,
$$
is a bijection.
\end{theo}
Here, $g\Gc_x=\{gk\;|\; k\in\Gc_x\}$ denotes a left coset.
The quotient set $\Gc/\Gc_x$ is not necessarily a group, since $\Gc_x$
need not be a normal subgroup of $\Gc$.

\noindent{\em Proof of Theorem~\ref{ebc21.03.26-t1}.}
\begin{enumerate}[label={(\alph*)}]
\item The mapping $\phi$ is well defined: if $g'\in g\Gc_x$ (that is,
if $g'$ is another representative of the same coset $g\Gc_x$), then
$g'x=gx$. Indeed, since $g'=gk$ for some $k\in\Gc_x$,
$$
g'x=(gk)x=g(kx)=gx.
$$

\item $\phi$ is injective. If $\phi(g\Gc_x)=\phi(g'\Gc_x)$, then
$$
g'x=gx
\;\Rightarrow\;
(g^{-1}g')x=x
\;\Rightarrow\;
g^{-1}g'\in\Gc_x.
$$
Hence $g'=gk$ for some $k\in\Gc_x$, and therefore
$g\Gc_x=g'\Gc_x$.

\item $\phi$ is surjective. If $x'\in O_x$, then there exists
$g\in\Gc$ such that $x'=gx$. Thus,
$$
\phi(g\Gc_x)=gx=x'.
$$
\hfill$\blacksquare$
\end{enumerate}
\begin{cor}[Orbit-stabilizer theorem]\label{ebc21.04.26-c1}
If the finite group $\Gc$ acts on the finite set $\Xc$, then for any
$x\in\Xc$,
$$
|\Gc|=|\Gc_x||O_x|.
$$
Hence $|O_x|$ divides $|\Gc|$, denoted by
$|O_x| \mid |\Gc|$.
\end{cor}
Indeed, for fixed $x\in\Xc$, Theorem~\ref{ebc21.03.26-t1} implies that
$|\Gc/\Gc_x|=|O_x|$, while
$$
|\Gc/\Gc_x|={|\Gc|\over |\Gc_x|}.
$$

\subsubsection{Burnside lemma}

\begin{lem}\label{ebc21.03.26-l1}
Let $\Gc$ be a group acting on a set $\Xc$. The stabilizers of all elements
belonging to the same orbit have the same order (that is, the same number of
elements if $\Gc$ is finite). In other words, if $x_2\in O_{x_1}$, then
$$
|\Gc_{x_1}|=|\Gc_{x_2}|.
$$
\end{lem}

\noindent{\em Proof.} If $x_2\in O_{x_1}$, there exists $g\in\Gc$ such that $x_2=g x_1$.  
Assume $h\in\Gc_{x_1}$. Then $g h g^{-1}\in\Gc_{x_2}$, since
$$
(g h g^{-1})x_2=g h g^{-1}(gx_1)=gh(g^{-1}g)x_1=
(gh)x_1=g(hx_1)=gx_1=x_2,
$$
i.e., $ghg^{-1}$ fixes $x_2$.  
This naturally leads us to define a map
$
\psi : \Gc_{x_1}\rightarrow\Gc_{x_2}
$
as
$$
\psi(h)=ghg^{-1}\in\Gc_{x_2}\quad\mbox{for all $h\in\Gc_{x_1}$.}
$$
We now prove that $\psi$ is a bijective group homomorphism, i.e., a map that
preserves the group operation:
\begin{enumerate}[label={(\alph*)}]
\item $\psi$ is a homomorphism. Indeed,  if $h,k\in\Gc_{x_1}$,
$$
\psi(hk)=g(hk)g^{-1}=gh(g^{-1}g)kg^{-1}=
(ghg^{-1})(gkg^{-1})=\psi(h)\psi(k).
$$

\item $\psi$ is injective. Since $\psi$ is a group homomorphism, it suffices to prove that $\ker\psi=\{e\}$. Indeed, if $h\in\ker\psi$, then
$$
\psi(h)=e\quad\Rightarrow\quad ghg^{-1}=e
\quad\Rightarrow\quad h=e.
$$

\item $\psi$ is surjective. Take any $k\in\Gc_{x_2}$, and consider $h=g^{-1}k g$. Then $h\in\Gc_{x_1}$, since
$$
hx_1=g^{-1}k g (g^{-1}x_2)=g^{-1}x_2=x_1,
$$
and
$$
\psi(h)=g(g^{-1}kg)g^{-1}=(gg^{-1})k(gg^{-1})=k.
$$
\end{enumerate}

Since $\psi$ is a bijection, we have $|\Gc_{x_1}|=|\Gc_{x_2}|$.\hfill$\blacksquare$

\begin{theo}[Burnside]\label{ebc29.04.26-th1}
Let $\Gc$ be a finite group acting on a finite set $\Xc$. Then the number
of distinct orbits is given~by
$$
N=\left|\Xc/\Gc\right|
=
{1\over|\Gc|}
\sum_{g\in\Gc} |\Xc_g|.
$$
\end{theo}

\noindent{\em Proof.}
Let $x_j$, $j=1,2,\dots,N$, be representatives of the $N$ distinct orbits.
Then,
$$
\sum_{g\in\Gc}|\Xc_g|
=\sum_{g\in\Gc}\sum_{x\in \Xc} \delta_{gx,x}
=\sum_{x\in\Xc}|\Gc_x|
=\sum_{j=1}^N\sum_{x\in O_{x_j}}|\Gc_x|
=\sum_{j=1}^N |O_{x_j}||\Gc_{x_j}|
=\sum_{j=1}^N|\Gc|
=N|\Gc|.
$$
In the second equality, the order of summation over $g\in\Gc$ and $x\in\Xc$
is exchanged, and the Kronecker delta is used to express the cardinalities
of the fixed-point set and stabilizer as
$|\Xc_g|=\sum_{x\in\Xc}\delta_{gx,x}$ and
$|\Gc_x|=\sum_{g\in\Gc}\delta_{gx,x}$.
The third equality follows from the fact that distinct orbits partition
$\Xc$, the fourth from Lemma~\ref{ebc21.03.26-l1}, and the fifth from~Corollary~\ref{ebc21.04.26-c1}. \hfill$\blacksquare$


\section{Conjugation and conjugacy classes}

\subsection{Conjugation in arbitrary groups}

Let $\Gc$ act on itself by {\em conjugation}:
$$
\varphi(g,h)=ghg^{-1}.
$$
In this section, we avoid the shorthand notation $gh:=ghg^{-1}$
introduced in Sec.~\ref{action_sss}, since it would be misleading, or even
inconsistent, in the present context. Although both $g$ and $h$ are group
elements, they play different roles in this expression: $g$ labels the group
element generating the action, whereas $h$ denotes the point in the domain on
which the action acts.

Conjugation defines a group action since:
\begin{enumerate}[label={(\alph*)}]
\item $\varphi(e,h)=ehe^{-1}=h$, for all $h\in\Gc$.

\item $\varphi(g_1,\varphi(g_2,h))
=\varphi(g_1,g_2hg_2^{-1})
=g_1(g_2hg_2^{-1})g_1^{-1}
=(g_1g_2)h(g_1g_2)^{-1}
=\varphi(g_1g_2,h)$,
for all $g_1,g_2,h\in\Gc$.
\end{enumerate}

The orbits $O_h$, $h\in\Gc$, under conjugation are called
{\em conjugacy classes} and are denoted $C_h$ instead of $O_h$.
Theorem~\ref{ebc29.04.26-th1} relates their sizes to the stabilizers
$\stab(h)=\Gc_h$, known in this context as {\em centralizers},
$$
\cen(h)=\{g\in\Gc: gh=hg\}
=\{g\in\Gc : \varphi(g,h)=h\}
=\Gc_h.
$$
In particular, $\cen(h)$ is a subgroup of $\Gc$. Therefore,
\begin{equation}
|C_h|
=
{|\Gc|\over |\cen(h)|}.
\label{|C_x|}
\end{equation}

\subsection{\boldmath Conjugacy classes in $S_n$}

The {\em symmetric group}~$S_n$ is the group of all permutations of~$n$
objects. Its elements can be represented as bijections
$\sigma:\{0,1,2,\dots,n-1\}\to\{0,1,2,\dots,n-1\}$, with group operation
given by composition. Applying first~$\sigma$ and then~$\tau$ to
$x\in\{0,1,\dots,n-1\}$ yields the permutation~$\tau\sigma$, defined by
$(\tau\sigma)(x)=\tau(\sigma(x))$. The order of~$S_n$ is
$|S_n|=n!$.

\subsubsection{Cycle and disjoint-cycle structure of $S_n$}

A {\em cycle} $(a_0,a_1,\dots,a_{k-1})$, with
$a_j\in\{0,1,\dots,n-1\}$, permutes the elements cyclically:
$$
\mbox{$a_0\to a_1$, $a_1\to a_2$, $\dots$,
$a_{k-1}\to a_0$}, 
$$
leaving all other elements fixed. More precisely,
its action is
$$
(a_0,a_1,\dots,a_{k-1})(x)=
\begin{cases}
a_{j+1\!\!\!\pmod{k}}
&
\mbox{if $x=a_j$ for some $j\in\{0,1,\dots,k-1\}$,}
\\
x
&
\mbox{if $x\notin\{a_0,a_1,\dots,a_{k-1}\}$.}
\end{cases}
$$
\begin{prop}
Any permutation can be written uniquely (up to order) as a product of {\em disjoint cycles}, i.e., cycles with no elements in common.
\end{prop}
\noindent{\em Proof:}
Let $\sigma \in S_n$. Consider the sequence
$$
0,\, \sigma(0),\, \sigma^2(0),\, \ldots
$$
Since $\{0,1,\ldots,n-1\}$ is finite, there exist $i<j$ such that $\sigma^i(0)=\sigma^j(0)$. Applying $\sigma^{-i}$, we obtain $\sigma^{j-i}(0)=0$. Let $r$ be the smallest positive integer such that $\sigma^r(0)=0$. Then
$$
(0,\, \sigma(0),\, \ldots,\, \sigma^{r-1}(0))
$$
is a cycle.

Now choose $a \in \{0,1,\ldots,n-1\}$ not belonging to this cycle and repeat the construction:
$$
a,\sigma(a),\sigma^2(a),\,\dots
$$
Since orbits are disjoint, so is the resulting cycle $(a,\sigma(a),\,\dots,\sigma^{s-1}(a))$. Since the set is finite, the process terminates after finitely many steps. Thus $\sigma$ can be written as a product of disjoint cycles.\hfill $\blacksquare$

\noindent Hereafter, we will assume that $\sigma$ is composed of $m_1$
cycles of length $\lambda_1$, $m_2$ cycles of length $\lambda_2$, and so on.
This cycle structure can be represented graphically by a Young diagram: a
left-justified arrangement of boxes consisting of $m_p$ rows of length
$\lambda_p$, for each $p$.

\subsubsection{Characterizing  the conjugacy classes of $S_n$}

We first recall the following elementary fact. 
\begin{prop}
Let $\pi\in S_n$ and let $\tau=(a_0,a_1,\dots,a_{k-1})$, $k\le n$. Then
\begin{equation}
\pi\tau\pi^{-1}
=
\left(\pi(a_0),\pi(a_1),\dots,\pi(a_{k-1})\right).
\label{ebc29.04.26-1}
\end{equation}
\end{prop}
\noindent{\em Proof:} Let $x\in\{0,1,\cdots,n-1\}$. We need to prove that
$$
(\pi\tau\pi^{-1})(x)=\left(\pi(a_0),\pi(a_1),\dots,\pi(a_{k-1})\right)(x).
$$
\begin{enumerate}[label={(\alph*)}]
\item Assume first that
$x\notin\{\pi(a_0),\dots,\pi(a_{k-1})\}$, so that
$(\pi(a_0),\pi(a_1),\dots,\pi(a_{k-1}))(x)=x$.
Under this assumption,
$\pi^{-1}(x)\notin\{a_0,\dots,a_{k-1}\}$, and hence
$$
\tau(\pi^{-1}(x))=\pi^{-1}(x),
$$
so
$$
(\pi\tau\pi^{-1})(x)
=
\pi(\tau(\pi^{-1}(x)))
=
\pi(\pi^{-1}(x))
=
x.
$$

\item If $x=\pi(a_j)$ for some $j$, then $\pi^{-1}(x)=a_j$, and
$$
(\pi\tau\pi^{-1})(x)
=
\pi(\tau(a_j))
=
\pi(a_{j+1}),
$$
where indices are understood modulo $k$. This coincides with
$(\pi(a_0),\dots,\pi(a_{k-1}))(x)$. \hfill$\blacksquare$
\end{enumerate}

It follows that if a permutation has the form $\tau_1\tau_2\cdots \tau_s$, where the $\tau_i$ are disjoint cycles, then
\begin{align}
\pi \tau_1\tau_2\cdots \tau_s \pi^{-1}&
= (\pi \tau_1\pi^{-1})(\pi\tau_2\pi^{-1})\cdots (\pi\tau_s \pi^{-1})\nonumber\\
&= \tau'_1\tau'_2\cdots \tau'_s,
\label{ebc30.04.26-4}
\end{align}
where each $\tau'_j$ is a cycle of the same length as $\tau_j$ (up to relabeling). Therefore, {\em conjugation preserves the cycle structure of a permutation}.

The converse is also true: {\em if two permutations have the same cycle structure, they are conjugate to one another}. To prove this statement, it suffices to consider two cycles of the same length $(a_1,a_2,\dots,a_k)$ and $(b_1,b_2,\dots,b_k)$, and define $\pi$ by $\pi(a_i)=b_i$. Then Eq.~(\ref{ebc29.04.26-1}) shows that $\pi\tau\pi^{-1}=\tau'$. Since permutations decompose into disjoint cycles and conjugation acts multiplicatively [cf.~Eq.~(\ref{ebc30.04.26-4})], the claim follows.

The above shows that 
\begin{prop}
Conjugacy classes of $S_n$ are uniquely determined by a partition
(equivalently, a Young diagram)
$\lambda=[\lambda_1^{m_1},\lambda_2^{m_2},\dots,\lambda_s^{m_s}]$,
where $\lambda_1>\lambda_2>\cdots>\lambda_s$ and
$\sum_p m_p\lambda_p=n$. We denote this by $\lambda\vdash n$.
\end{prop}

\subsubsection{Size of conjugacy classes of $S_n$}

For the symmetric group $S_n$, Eq.~(\ref{|C_x|}) yields
$$
|C_\sigma|={n!\over |\cen(\sigma)|}, \quad \sigma\in S_n.
$$
%
Now let us compute $|\cen(\sigma)|$. Assume $\pi\in\cen(\sigma)$ and write
$\sigma=\tau_1\tau_2\cdots\tau_s$, where the cycles $\tau_p$ are disjoint.
The condition $\pi\sigma\pi^{-1}=\sigma$ implies
$$
(\pi\tau_1\pi^{-1})(\pi\tau_2\pi^{-1})\cdots(\pi\tau_s\pi^{-1})
=\tau_1\tau_2\cdots\tau_s.
$$
Thus conjugation by $\pi$ must permute cycles of equal length. If the distinct
cycle lengths are $\lambda_j$, with multiplicities $m_j$, then cycles of
length $\lambda_j$ can be permuted in $m_j!$ ways, and each such cycle admits
$\lambda_j$ cyclic rotations. Hence
$$
|\cen(\sigma)|=\prod_j m_j!\,\lambda_j^{m_j}.
$$
Finally,
\begin{equation}
|C_\lambda|:=|C_\sigma|
=
{n!\over \prod_j m_j!\,\lambda_j^{m_j}}.
\label{|Clambda|}
\end{equation}

\newpage

\section{A worked-out example: the cyclic group}

We illustrate the general construction of the encoding states
$|u^\mu_\alpha\rangle$ for the cyclic group with the simplest nontrivial
example. Consider a memory device consisting of a ring of $n=4$ qubits, for
instance four cold atoms with $d=2$ internal states. If the atoms remain
perfectly localized and individually identifiable, the device can store the
full set of $2^4=16$ binary configurations. However, if coherence is lost and
the atoms behave effectively as classical balls or beads of two colors,
configurations related by cyclic drifts of the atoms become operationally
indistinguishable. The number of distinguishable messages is then equal to
the number of distinct necklaces that can be formed with those beads.

The relevant group, $\Gc=\langle r\rangle\cong\mathbb Z_4$, is generated by
$r=(0,1,2,3)$, representing a one-step clockwise drift of the atoms along the
ring. We use the standard cycle notation in which the content of site $i$ is
moved to site $i+1\pmod 4$. With the same notation, the remaining three elements are
$r^2=(0,2)(1,3)$, $r^3=(0,3,2,1)$, and
$r^4=e=(0)(1)(2)(3)$. Recalling that $c(\sigma)$ denotes the number
of cycles in the disjoint-cycle decomposition of $\sigma$, we have
$c(r)=c(r^3)=1$, $c(r^2)=2$, and $c(e)=4$. Figure~\ref{cycle_f} illustrates the action
of the four group elements.

\begin{figure*}[h]
\centering
\begin{minipage}[c]{0.23\textwidth}
\centering
\begin{tikzpicture}[>=stealth, scale=0.75]
\node[circle, draw, minimum size=4mm] (1) at (0,2) {$0$};
\node[circle, draw, minimum size=4mm] (2) at (2,2) {$1$};
\node[circle, draw, minimum size=4mm] (3) at (2,0) {$2$};
\node[circle, draw, minimum size=4mm] (4) at (0,0) {$3$};

\draw[->, bend left=12, shorten >=2mm, shorten <=2mm] (1) to node[above] {$r$} (2);
\draw[->, bend left=12, shorten >=2mm, shorten <=2mm] (2) to node[right] {$r$} (3);
\draw[->, bend left=12, shorten >=2mm, shorten <=2mm] (3) to node[below] {$r$} (4);
\draw[->, bend left=12, shorten >=2mm, shorten <=2mm] (4) to node[left] {$r$} (1);
\end{tikzpicture}
\end{minipage}%
\begin{minipage}[c]{0.23\textwidth}
\centering
\begin{tikzpicture}[>=stealth, scale=0.75]
\node[circle, draw, minimum size=4mm] (1) at (0,2) {$0$};
\node[circle, draw, minimum size=4mm] (2) at (2,2) {$1$};
\node[circle, draw, minimum size=4mm] (3) at (2,0) {$2$};
\node[circle, draw, minimum size=4mm] (4) at (0,0) {$3$};

\draw[->, bend left=18, shorten >=2mm, shorten <=2mm] (1) to node[pos=0.85, above right] {$r^2$} (3);
\draw[->, bend left=18, shorten >=2mm, shorten <=2mm] (3) to node[pos=0.85, below left] {$r^2$} (1);
\draw[->, bend left=18, shorten >=2mm, shorten <=2mm] (2) to node[pos=0.85, below right] {$r^2$} (4);
\draw[->, bend left=18, shorten >=2mm, shorten <=2mm] (4) to node[pos=0.85, above left] {$r^2$} (2);
\end{tikzpicture}
\end{minipage}%
\begin{minipage}[c]{0.23\textwidth}
\centering
\begin{tikzpicture}[>=stealth, scale=0.75]
\node[circle, draw, minimum size=4mm] (1) at (0,2) {$0$};
\node[circle, draw, minimum size=4mm] (2) at (2,2) {$1$};
\node[circle, draw, minimum size=4mm] (3) at (2,0) {$2$};
\node[circle, draw, minimum size=4mm] (4) at (0,0) {$3$};

\draw[->, bend right=12, shorten >=2mm, shorten <=2mm] (2) to node[above] {$r^3$} (1);
\draw[->, bend right=12, shorten >=2mm, shorten <=2mm] (3) to node[right] {$r^3$} (2);
\draw[->, bend right=12, shorten >=2mm, shorten <=2mm] (4) to node[below] {$r^3$} (3);
\draw[->, bend right=12, shorten >=2mm, shorten <=2mm] (1) to node[left] {$r^3$} (4);
\end{tikzpicture}
\end{minipage}%
\begin{minipage}[c]{0.23\textwidth}
\centering
\begin{tikzpicture}[>=stealth, scale=0.75]
\node[circle, draw, minimum size=4mm] (1) at (0,2) {$0$};
\node[circle, draw, minimum size=4mm] (2) at (2,2) {$1$};
\node[circle, draw, minimum size=4mm] (3) at (2,0) {$2$};
\node[circle, draw, minimum size=4mm] (4) at (0,0) {$3$};

\draw[->, shorten >=1.5mm, shorten <=1.5mm]
(1) edge[loop right, in=35, out=-35, looseness=8]
node[left] {$r^4$} (1);

\draw[->, shorten >=1.5mm, shorten <=1.5mm]
(2) edge[loop below, in=-55, out=-125, looseness=8]
node[above] {$r^4$} (2);

\draw[->, shorten >=1.5mm, shorten <=1.5mm]
(3) edge[loop left, in=-145, out=145, looseness=8]
node[right] {$r^4$} (3);

\draw[->, shorten >=1.5mm, shorten <=1.5mm]
(4) edge[loop above, in=125, out=55, looseness=8]
node[below] {$r^4$} (4);

\end{tikzpicture}
\end{minipage}

\caption{Action of the elements of $\mathbb Z_4$ on four sites arranged on a ring.}
\label{cycle_f}
\end{figure*}
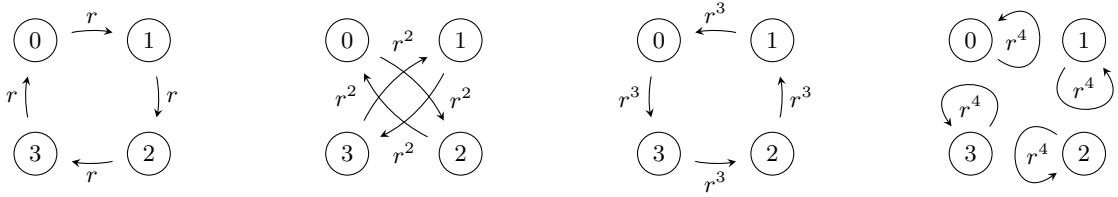

Using P\'olya's enumeration theorem, Eq.~(\ref{burnside}), the number of
classically distinguishable messages is
$$
N_{\rm c}
=
\left|X/\Gc\right|
=
{1\over4}\sum_{k=0}^{3}2^{c(r^k)}
=
{1\over 4}\left(2^4+2^1+2^2+2^1\right)
=
6 .
$$
Thus, the $16$ binary strings used to encode the messages collapse into six
equivalence classes under cyclic rotations. These correspond to the six
binary necklaces shown in Fig.~\ref{necklaces_f}, with lexicographically
ordered representatives
$$
\xbf_1=0000,\quad \xbf_2=0001,\quad \xbf_3=0011,\quad
\xbf_4=0101,\quad \xbf_5=0111,\quad \xbf_6=1111.
$$
\begin{figure}[h]
\centering

\newcommand{\necklace}[5]{%
\begin{tikzpicture}[scale=0.7, baseline=-0.5ex]

\draw (-1,1)--(1,1)--(1,-1)--(-1,-1)--cycle;

\node[circle, draw, fill=#2, minimum size=1.4mm, inner sep=0pt] at (-1,1) {};
\node[circle, draw, fill=#3, minimum size=1.4mm, inner sep=0pt] at (1,1) {};
\node[circle, draw, fill=#4, minimum size=1.4mm, inner sep=0pt] at (1,-1) {};
\node[circle, draw, fill=#5, minimum size=1.4mm, inner sep=0pt] at (-1,-1) {};

\node at (0,-1.35) {\scriptsize $#1$};

\path[use as bounding box] (-1.15,-1.45) rectangle (1.15,1.15);

\end{tikzpicture}%
}

\makebox[\textwidth][c]{%
\necklace{0000}{white}{white}{white}{white}
\hspace{0.45cm}
\necklace{0001}{white}{white}{white}{black}
\hspace{0.45cm}
\necklace{0011}{white}{white}{black}{black}
\hspace{0.45cm}
\necklace{0101}{white}{black}{white}{black}
\hspace{0.45cm}
\necklace{0111}{white}{black}{black}{black}
\hspace{0.45cm}
\necklace{1111}{black}{black}{black}{black}
}

\caption{The six binary necklace configurations with four beads, where $\circ=0$ and $\bullet=1$.}
\label{necklaces_f}
\end{figure}
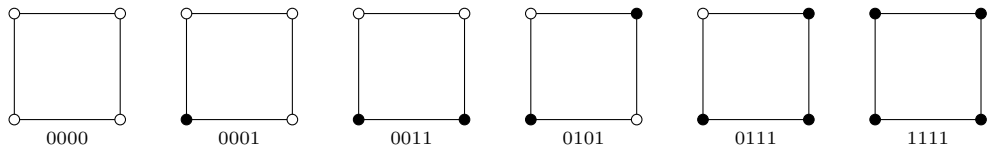

We now turn to the quantum setting. If coherent superpositions can be prepared
and preserved, the vectors $|\tilde u^k_j\rangle$ in Eq.~(\ref{ebc26.03.26-1}) define a natural
Fourier basis for the subspaces $\Span(O_{\xbf_j})$. The resulting
unnormalized states are shown in Table~\ref{ebc01.04.26-t1}. They diagonalize
the action of $r$, transforming by multiplicative phases which are precisely
the irreducible characters of $\mathbb Z_4$.

\begin{table}[h]
\centering
\scriptsize
\setlength{\tabcolsep}{8pt}
\renewcommand{\arraystretch}{1.5}

\begin{tabular}{|c|c|c|c|}
\hline
Orbit \# $j$
&
Orbit representative
&
Fourier basis elements of $\Span(O_{\xbf_j})$: $|\tilde u_j^k\rangle$
&
Irrep character
\\
\hline

$1$
&
$0000$
&
$|0000\rangle$
&
$1$
\\
\hline

$2$
&
$0001$
&
\begin{tabular}[c]{@{}c@{}}
$|0001\rangle+|1000\rangle+|0100\rangle+|0010\rangle$
\\
$|0001\rangle-i|1000\rangle-|0100\rangle+i|0010\rangle$
\\
$|0001\rangle-|1000\rangle+|0100\rangle-|0010\rangle$
\\
$|0001\rangle+i|1000\rangle-|0100\rangle-i|0010\rangle$
\end{tabular}
&
\begin{tabular}[c]{@{}c@{}}
$1$
\\
$i$
\\
$-1$
\\
$-i$
\end{tabular}
\\
\hline

$3$
&
$0011$
&
\begin{tabular}[c]{@{}c@{}}
$|0011\rangle+|1001\rangle+|1100\rangle+|0110\rangle$
\\
$|0011\rangle-i|1001\rangle-|1100\rangle+i|0110\rangle$
\\
$|0011\rangle-|1001\rangle+|1100\rangle-|0110\rangle$
\\
$|0011\rangle+i|1001\rangle-|1100\rangle-i|0110\rangle$
\end{tabular}
&
\begin{tabular}[c]{@{}c@{}}
$1$
\\
$i$
\\
$-1$
\\
$-i$
\end{tabular}
\\
\hline

$4$
&
$0101$
&
\begin{tabular}[c]{@{}c@{}}
$|0101\rangle+|1010\rangle$
\\
$|0101\rangle-|1010\rangle$
\end{tabular}
&
\begin{tabular}[c]{@{}c@{}}
$1$
\\
$-1$
\end{tabular}
\\
\hline

$5$
&
$0111$
&
\begin{tabular}[c]{@{}c@{}}
$|0111\rangle+|1011\rangle+|1101\rangle+|1110\rangle$
\\
$|0111\rangle-i|1011\rangle-|1101\rangle+i|1110\rangle$
\\
$|0111\rangle-|1011\rangle+|1101\rangle-|1110\rangle$
\\
$|0111\rangle+i|1011\rangle-|1101\rangle-i|1110\rangle$
\end{tabular}
&
\begin{tabular}[c]{@{}c@{}}
$1$
\\
$i$
\\
$-1$
\\
$-i$
\end{tabular}
\\
\hline

$6$
&
$1111$
&
$|1111\rangle$
&
$1$
\\
\hline

\end{tabular}

\caption{Unnormalized Fourier basis vectors $|\tilde u_j^k\rangle$ for~\mbox{$d=2$}
and~$n=4$, grouped according to the orbits of the cyclic action.}
\label{ebc01.04.26-t1}
\end{table}

Collecting the Fourier vectors with the same character yields the encoding
states $|u^\mu_\alpha\rangle$ shown in Table~\ref{ebc01.04.26-t2}. Here
$\mu\in\{1,-i,-1,i\}$ labels the irreducible representations of
$\mathbb Z_4$ through the character value on the generator $r$, while
$\alpha$ labels the corresponding multiplicities. One finds
$m_1=6$, $m_{-1}=4$, and $m_{\pm i}=3$, so that
$$
N_{\rm q}
=
m_1+m_{-i}+m_{-1}+m_i
=
6+3+4+3
=
16
=
2^4 .
$$
Thus, coherent encoding completely removes the effect of cyclic positional
uncertainty, restoring the identity-channel value. This finite example
illustrates the general result $N_{\rm q}=d^n$,
Eq.~(\ref{cyclic_general}), for cyclic permutation uncertainty.
Decoding is performed by projection onto the same Fourier basis
$|u^\mu_\alpha\rangle$.

\newpage

\begin{table}[h]
\centering
\scriptsize
\setlength{\tabcolsep}{10pt}
\renewcommand{\arraystretch}{1.5}

\begin{tabular}{|c|c|c|}
\hline
$\mu$ & $\alpha$ & $|u^\mu_\alpha\rangle$ \\
\hline

\multirow{6}{*}{$1$}
& $0$ & $|0000\rangle$ \\
\cline{2-3}
& $1$ & $|0001\rangle+|1000\rangle+|0100\rangle+|0010\rangle$ \\
\cline{2-3}
& $2$ & $|0011\rangle+|1001\rangle+|1100\rangle+|0110\rangle$ \\
\cline{2-3}
& $3$ & $|0101\rangle+|1010\rangle$ \\
\cline{2-3}
& $4$ & $|0111\rangle+|1011\rangle+|1101\rangle+|1110\rangle$ \\
\cline{2-3}
& $5$ & $|1111\rangle$ \\
\hline

\multirow{3}{*}{$i$}
& $0$ & $|0001\rangle-i|1000\rangle-|0100\rangle+i|0010\rangle$ \\
\cline{2-3}
& $1$ & $|0011\rangle-i|1001\rangle-|1100\rangle+i|0110\rangle$ \\
\cline{2-3}
& $2$ & $|0111\rangle-i|1011\rangle-|1101\rangle+i|1110\rangle$ \\
\hline

\multirow{4}{*}{$-1$}
& $0$ & $|0001\rangle-|1000\rangle+|0100\rangle-|0010\rangle$ \\
\cline{2-3}
& $1$ & $|0011\rangle-|1001\rangle+|1100\rangle-|0110\rangle$ \\
\cline{2-3}
& $2$ & $|0101\rangle-|1010\rangle$ \\
\cline{2-3}
& $3$ & $|0111\rangle-|1011\rangle+|1101\rangle-|1110\rangle$ \\
\hline

\multirow{3}{*}{$-i$}
& $0$ & $|0001\rangle+i|1000\rangle-|0100\rangle-i|0010\rangle$ \\
\cline{2-3}
& $1$ & $|0011\rangle+i|1001\rangle-|1100\rangle-i|0110\rangle$ \\
\cline{2-3}
& $2$ & $|0111\rangle+i|1011\rangle-|1101\rangle-i|1110\rangle$ \\
\hline

\end{tabular}

\caption{Unnormalized encoding states $|u^\mu_\alpha\rangle$ for~\mbox{$d=2$}
and~$n=4$, grouped according to the irrep label~$\mu$ of $\mathbb Z_4$.}
\label{ebc01.04.26-t2}

\end{table}

\newpage

\section{Proof of Eq.~(\ref{Nq polya})}

In this section we prove Eq.~(\ref{Nq polya}), which gives the number
$N_{\rm q}$ of classical messages that can be transmitted or stored using
ancilla-free quantum encoding when the symmetry group $\Gc$ is totally
orthogonal. The starting point is Eq.~(\ref{M_q}), valid in full generality:
\begin{equation}
N_{\rm q}=\sum_\mu m_\mu ,
\label{N_q}
\end{equation}
where $m_\mu$ denotes the multiplicity of irrep $\mu$ in the decomposition of
the total Hilbert space $(\mathbb{C}^d)^{\otimes n}$. From character theory,
each multiplicity is given by
$$
m_\mu
=
\langle \chi,\chi_\mu\rangle
=
\frac{1}{|\Gc|}
\sum_{\sigma\in\Gc}
\chi(\sigma)\,
\overline{\chi_\mu(\sigma)}
=
\frac{1}{|\Gc|}
\sum_{\sigma\in\Gc}
|\Xc_\sigma|\,
\overline{\chi_\mu(\sigma)} ,
$$
where the bar denotes complex conjugation. The second equality follows from
the definition of the inner product of
characters~$\langle\chi,\chi_\mu\rangle$. In the third equality we used
$$
\chi(\sigma)
=
\tr U(\sigma)
=
\sum_{\xbf\in X}
\langle\xbf|U(\sigma)|\xbf\rangle
=
|\Xc_\sigma| ,
$$
since the trace counts the basis states left invariant by the action of
$\sigma$.
Equation~(\ref{N_q}) then becomes
\begin{equation}
N_{\rm q}
=
\frac{1}{|\Gc|}
\sum_{\sigma\in\Gc}
|\Xc_\sigma|
\sum_\mu
\chi_\mu(\sigma),
\label{david_2}
\end{equation}
where we used that totally orthogonal groups have all irreps realizable
over~$\mathbb R$, and exchanged the sums over $\mu$ and~$\sigma$. The inner sum admits a simple
interpretation given by the following lemma.

\begin{lem}\label{david_1}
Let $\Gc$ be a finite totally orthogonal group. Then
$$
\sum_\mu\chi_\mu(\sigma)=n(\sigma),
$$
where $n(\sigma)$ denotes the number of square roots of~$\sigma$, and the sum runs over all irreducible characters $\chi_\mu$ of~$\Gc$.
\end{lem}
In other words,
$
n(\sigma)=|\Ss_\sigma|,
$
where
$
\Ss_\sigma:=\{\tau\in\Gc\;|\;\tau^2=\sigma\}
$
is the set of square roots of~$\sigma$.
The proof of this lemma uses the following proposition.

\begin{prop}
The number of square roots is a class function. Namely, if $\sigma$ and
$\sigma'$ belong to the same conjugacy class, then
$
n(\sigma)=n(\sigma').
$
\end{prop}
\noindent{\em Proof.}
Assume $\sigma$ and $\sigma'$ belong to the same conjugacy class, so that
$\sigma'=\pi\sigma\pi^{-1}$ for some $\pi\in\Gc$. Define
$
\psi:\Ss_\sigma\longrightarrow\Ss_{\sigma'}
$
by
$
\psi(\tau)=\pi\tau\pi^{-1}.
$
This map is well defined since, for every $\tau\in\Ss_\sigma$,
\begin{equation}
\psi(\tau)^2
=
(\pi\tau\pi^{-1})^2
=
\pi\tau^2\pi^{-1}
=
\pi\sigma\pi^{-1}
=
\sigma'.
\label{psi^2}
\end{equation}
Hence $\psi(\tau)\in\Ss_{\sigma'}$.
Moreover, $\psi$ is bijective. Indeed, its inverse is
$
\psi^{-1}(\tau')=\pi^{-1}\tau'\pi,
$
which maps $\Ss_{\sigma'}$ back to $\Ss_\sigma$ by the same argument used in
Eq.~(\ref{psi^2}). Hence,
$
|\Ss_\sigma|=|\Ss_{\sigma'}|,
$
or equivalently,
$
n(\sigma)=n(\sigma').
$
\hfill$\blacksquare$

\noindent{\em Proof of Lemma~\ref{david_1}.}
Since $n(\sigma)$ is a class function, it can be expanded in the basis of irreducible characters:
$$
n(\sigma)
=
\sum_\mu
\langle n,\chi_\mu\rangle\chi_\mu(\sigma).
$$
The coefficients of this expansion are
$$
\langle n,\chi_\mu\rangle
=
\frac{1}{|\Gc|}
\sum_{\sigma\in\Gc}
n(\sigma)\chi_\mu(\sigma)
=
\frac{1}{|\Gc|}
\sum_{\sigma\in\Gc}
\sum_{\tau\in\Gc}
\delta_{\sigma,\tau^2}\chi_\mu(\sigma)
=
\frac{1}{|\Gc|}
\sum_{\tau\in\Gc}
\chi_\mu(\tau^2)
=
\nu_\mu,
$$
where $\nu_\mu$ denotes the Frobenius--Schur indicator of the irrep~$\mu$~\cite{james2001representations}, and we used the Kronecker delta to represent $n(\sigma)$ as
\begin{equation}
n(\sigma)
=
\sum_{\tau\in\Gc}\delta_{\sigma,\tau^2}.
\label{delta_rep}
\end{equation}
Since $\Gc$ is totally orthogonal, all irreducible representations are realizable over~$\mathbb R$, and therefore $\nu_\mu=+1$ for all~$\mu$. Hence,
$$
\langle n,\chi_\mu\rangle=1,
$$
which proves the lemma. \hfill$\blacksquare$

We have used the well-known fact~\cite{james2001representations} that the
Frobenius--Schur indicator $\nu_\mu$ takes the values $+1$, $0$, or $-1$
depending on whether the irreducible representation~$\mu$ is respectively
realizable over $\mathbb R$, genuinely complex, or quaternionic
(pseudoreal).

Lemma~\ref{david_1} allows Eq.~(\ref{david_2}) to be rewritten as
$$
N_{\rm q}
=
\frac{1}{|\Gc|}
\sum_{\sigma\in\Gc}
|\Xc_\sigma|\, n(\sigma).
$$
Using again the representation of Eq.~(\ref{delta_rep}),
we obtain
$$
N_{\rm q}
=
\frac{1}{|\Gc|}
\sum_{\sigma\in\Gc}
|\Xc_\sigma|
\sum_{\tau\in\Gc}
\delta_{\sigma,\tau^2}
=
\frac{1}{|\Gc|}
\sum_{\tau\in\Gc}
|\Xc_{\tau^2}|.
$$
Equivalently,
$$
N_{\rm q}
=
\frac{1}{|\Gc|}
\sum_{\sigma\in\Gc}
|\Xc_{\sigma^2}|={1\over |\Gc|}
\sum_{\sigma\in\Gc} d^{c(\sigma^2)}
.
$$
The last step follows from the fact that a configuration remains invariant
under $\sigma$ precisely when all positions belonging to the same cycle of
$\sigma$ contain the same symbol. Since each cycle can be assigned any of the
$d$ symbols independently, the number of fixed strings is
$
|\Xc_\sigma|=d^{c(\sigma)}.
$

This expression should be compared with its classical counterpart, Eq.~(\ref{burnside}):
$$
N_{\rm c}
=
\frac{1}{|\Gc|}
\sum_{\sigma\in\Gc}
d^{c(\sigma)}.
$$
Since squaring group elements may increase the number of cycles in their
disjoint cycle decomposition, the sum may become
larger. In particular, elements satisfying $\sigma^2=e$ contribute the maximal
value $d^n$. Hence, one generally has~
$
N_{\rm q}\ge N_{\rm c}.
$


\newpage

\section{\boldmath The cycle index function and the cycle index series of $S_n$}

The {\em cycle index function of $\Gc\le S_n$}~\cite{vanLintWilson2001,polya2012combinatorial,Serre1977} is a polynomial in the variables $a_1,a_2,\dots,a_n$ defined as
$$
Z_\Gc(a_1,a_2,\dots,a_n):={1\over |\Gc|}\sum_{g\in\Gc}\prod_{k=1}^n a_k^{c_k(g)} ,
$$
where $c_k(g)$ is the number of cycles of length $k$ ($k$-cycles) of $g\in\Gc$.

Let us consider the cycle index function of the symmetric group $S_n$.
If we allow $m_k$ to take the value zero, we can write the partition 
$\lambda=[\lambda_1^{m_1},\lambda_2^{m_2},\dots,\lambda_s^{m_s}]$ as
$$
\lambda=[1^{m_1},2^{m_2},\dots,n^{m_n}],
$$
and the condition $\lambda\vdash n$ can be expressed as
\begin{equation}
\sum_{k=1}^n m_k k=n,\quad m_k\ge 0.
\label{ebc29.04.26-2}
\end{equation}
Then,
\begin{multline*}
Z_{S_n}(a_1,a_2,\dots,a_n)
={1\over n!}\sum_{\sigma\in S_n}\prod_{k=1}^n a_k^{m_k(\sigma)}=\\
{1\over n!}\sum_{\lambda\vdash n}\left(\sum_{\sigma\in C_\lambda}\prod_{k=1}^n a_k^{m_k(\sigma)}\right)
= {1\over n!}\sum_{\lambda\vdash n}|C_\lambda|\prod_{k=1}^n a_k^{m_k}=\\
{1\over n!}\!\!\sum_{\two{m_1,m_2,\dots,m_n \ge 0}{\sum_{k=1}^n m_k k=n}}
{n!\over \prod_{k=1}^n m_k! k^{m_k}} \prod_{k=1}^n a_k^{m_k} ,
\end{multline*}
where we have used Eq.~(\ref{|Clambda|}).
 
We next introduce the generating function of $Z_{S_n}$, known as the {\em cycle index series of $S_n$}:
\begin{equation}
\sum_{n\ge0} x^n Z_{S_n}(a_1,a_2,\dots,a_n)=
\exp\left(\sum_{k\ge1}{a_k x^k\over k}\right).
\label{ebc30.04.26-1}
\end{equation}
This equation must be understood in the following sense: the coefficient of $x^n$ in the formal power expansion of the right-hand side is precisely $Z_{S_n}(a_1,a_2,\dots,a_n)$. This is a purely algebraic relation, as is usual in combinatorics; issues of convergence do not arise here.

\noindent{\em Proof of Eq.~(\ref{ebc30.04.26-1})}
\begin{equation}
x^n Z_{S_n}(a_1,a_2,\dots,a_n)=
\hspace*{-1em}
\sum_{\two{m_1,m_2,\dots,m_n \ge 0}{\sum_{k=1}^n m_k k=n}}
{x^n\prod_{k=1}^n a_k^{m_k}\over \prod_{k=1}^n m_k! k^{m_k}} =
\hspace*{-1em}
\sum_{\two{m_1,m_2,\dots,m_n \ge 0}{\sum_{k=1}^n m_k k=n}}
\prod_{k=1}^n {1\over m_k!}
\left({a_k x^k \over k}\right)^{m_k},
\label{ebc30.04.26-2}
\end{equation}
where the last equality is justified by Eq.~(\ref{ebc29.04.26-2}):
$$
x^n=x^{\sum_{k=1}^n k m_k}
=\prod_{k=1}^n (x^k)^{m_k}.
$$
Now consider
\begin{equation}
\prod_{k\ge 1} \sum_{m_k\ge0} {1\over m_k!}\left({a_k x^k \over k}\right)^{m_k}.
\label{ebc30.04.26-3}
\end{equation}
For any integer $n\ge0$, the term of order $x^n$ is obtained from those choices of $(m_k)_{k\ge1}$ such that $\sum_{k\ge1} k m_k = n$. Since $k\ge1$, necessarily $m_k=0$ for all $k>n$. Hence, at fixed order $n$, the condition reduces to $\sum_{k=1}^n k m_k = n$, and the product $\prod_{k\ge1}$ reduces to $\prod_{k=1}^n$. Therefore, the contribution of order $x^n$ in Eq.~(\ref{ebc30.04.26-3}) coincides with Eq.~(\ref{ebc30.04.26-2}). Therefore,
$$
\sum_{n\ge0}x^nZ_{S_n}(a_1,a_2,\dots, a_n)=
\prod_{k\ge 1} \sum_{m_k\ge0} {1\over m_k!}\left({a_k x^k \over k}\right)^{m_k}=
\prod_{k\ge 1} \exp\left({a_k x^k\over k}\right)=
\exp\left(\sum_{k\ge1}{a_k x^k\over k}\right).
$$
This completes the proof.\hfill $\blacksquare$



\end{document}